\documentclass{amsproc}

\usepackage[T1]{fontenc}
\usepackage{multicol}

\usepackage{amsmath, amssymb, amsfonts, amsthm}
\usepackage[all]{xy}

\usepackage{xcolor}
\usepackage{hyperref}
\hypersetup{
    colorlinks,
    linkcolor={red!50!black},
    citecolor={green!50!black},
    urlcolor={blue!80!black}
}

\usepackage{tikz}
\usepackage{mathrsfs}

\theoremstyle{plain}
\newtheorem{theorem}{Theorem}[section]

\newtheorem{question}[theorem]{Question}
\newtheorem{proposition}[theorem]{Proposition}
\newtheorem{lemma}[theorem]{Lemma}
\newtheorem{corollary}[theorem]{Corollary}

\theoremstyle{definition}
\newtheorem{definition}[theorem]{Definition}
\theoremstyle{remark}
\newtheorem{remark}[theorem]{Remark}

\numberwithin{equation}{section}

\renewenvironment{proof}{{\it Proof.~}}{\qed}

\DeclareMathOperator{\Pic}{Pic}
\DeclareMathOperator{\Gal}{Gal}

\DeclareMathOperator{\elm}{elm}
\DeclareMathOperator{\Num}{Num}

\DeclareMathOperator{\Spec}{Spec}

\DeclareMathOperator{\Eff}{Eff}
\DeclareMathOperator{\PRS}{PRS}

\def\F{\mathbf{F}}

\def\P{\mathbf{P}}
\def\Z{\mathbf{Z}}
\def\SS{\mathbf{S}}

\def\C{\mathcal{C}}
\def\L{\mathscr{L}}
\def\M{\mathscr{M}}
\def\NN{\mathscr{N}}
\def\O{\mathcal{O}}
\def\E{\mathscr{E}}

\makeatletter
\@namedef{subjclassname@2020}{\textup{2020} Mathematics Subject Classification}
\makeatother

\title{Ruled surfaces over finite fields, and some codes over them}

\author[R. Blache]{Régis Blache}
\address{LAMIA, Université des Antilles}
\email{regis.blache@univ-antilles.fr}

\author[E. Hallouin]{Emmanuel Hallouin}
\address{Institut de Math\'ematiques de Toulouse, UMR 5219}
\email{hallouin@univ-tlse2.fr}

\date{\today} 

\thanks{This work was funded by the French Agence Nationale de la Recherche through ANR BARRACUDA
(ANR-21-CE39-0009-BARRACUDA).}

\subjclass[2020]{11G25, 14J26, 14G50}

\keywords{Ruled surfaces over finite fields, Segre invariants, evaluation codes}

\begin{document}

\begin{abstract}
In the first part of this article, we consider ruled surfaces defined over a finite field; we introduce invariants for them, and describe some explicit contructions that illustrate possible behaviour of these invariants. 

In the second part, we consider evaluation codes on some such surfaces; we first estimate their parameters, then we construct asymptotically good families of such codes, and we show that their asymptotic parameters are better than the ones of the corresponding product codes. We also consider local properties of these codes.
\end{abstract}

\maketitle


%
%
%

\section*{Introduction}

The first motivation for this article was the study of evaluation codes on ruled surfaces. In the course of doing this, we realized that the study of such surfaces, defined over a finite field, was very incomplete. To the best of our knowledge, there has been no work on this specific topic, except when the base curve has genus zero; then we get the famous Hirzebruch surfaces, and there is nothing new compared to the algebraically closed case.

For this reason, the objectives of this article have integrated the first steps of a study of ruled surfaces over finite fields and their invariants. Let us give a brief historical account, in order to recast our results.

Classically (in the works of Cayley, Chasles or Cremona) an algebraic ruled surface --defined over the complex numbers-- is a reduced and irreducible surface in $\P^3$ that is the union of the lines it contains. The approach by C. Segre \cite{segr} of considering such varieties as projections of other ones in higher dimensional projective spaces lead him to a general study of these surfaces and their properties. 

Segre's works seem to have been forgotten for a long time. Actually his study was renewed almost a century later using more modern tools of algebraic geometry, and some of his results were given a new proof, without reference to the original ones \cite{maru,naga}. Along this rediscovery, the definition of ruled surfaces evolved to a new one, equivalent to that of projective bundle associated to a locally free rank two sheaf on the base curve. This definition has become the standard one and can be found in many textbooks \cite{bade,hart}. In this way the study of ruled surfaces becomes equivalent to that of certain vector bundles on curves; this adds a new perspective, that most of the recent results on the subject involve primarily \cite{lana}. 

Note that our first result will be that, using a correct geometric definition of a ruled surface over a finite field, the description as a projective bundle remains correct in our setting. In turn, the general results on these objects will tell us many important geometric invariants of our original ruled surface (its canonical divisor, Picard group, intersection product...). Note also that, although our original interest lead us to privilege the vocabulary of ruled surfaces, the results here can be seen as results on rank two locally free sheaves on a curve defined over a finite field.  

Once this is done, we continue our study of ruled surfaces by introducing elementary transforms associated to a non rational point, and showing some of their properties. Recall that, over an algebraically closed field, ruled surfaces share with rational surfaces the property that they do not have a unique minimal model. The birational maps from one minimal model to another are compositions of elementary transforms (one blows up a point then contracts the strict transform of the line passing through it). Our main result about the new elementary transforms is that the above result remains true over a finite field: they are sufficient to obtain any ruled surface over a given curve from the simplest one (the product of the curve by the projective line). 

Associated to a ruled surface (or a rank two locally free sheaf) is its \emph{Segre invariant}: it is the minimum self-intersection number of a section, which is a curve isomorphic to the base curve drawn on the ruled surface. We define a new Segre invariant, that we call \emph{arithmetic}, by considering only the sections for which the isomorphism is defined over our base field. We compare it to the original Segre invariant, that we call \emph{geometric}, and we give bounds for it, comparable to Nagata's well-known bound (already discovered by Segre). We also provide some constructions of surfaces having given couples of invariants.

These results follow easily from classical results about curves (in particular the Riemann-Roch theorem), but are somewhat sporadic: except the case when the invariants are negative, we are far from describing all possible couples occuring for ruled surfaces over a given curve. For this reason, we feel that the subject deserves a more thorough treatment. 

The second part of the paper is devoted to our original concern; we construct evaluation codes, and explore their properties. Such codes are obtained from the global sections of some effective divisor, evaluated at the rational points of the surface. They have already been considered when the base curve is the projective line, ie over Hirzebruch surfaces \cite{cane,nard,svav}. But in this case there are very few different surfaces (actually one for each value of the Segre invariant -- the arithmetic and geometric ones are equal) and one can consider an exhaustive study of the evaluation codes. 

In the general case, this is no longer possible, and we restrict our attention to the codes which have the best parameters among those we have worked out. We do not pretend to have constructed the evaluation codes from ruled surfaces having the best parameters, even if the parameters are better than the ones given in \cite[Proposition 4.2]{hans}. This result is the only one we are aware of on this subject, but there the author only uses results on ruled surfaces that are already valid in the algebraically closed case (actually he mostly uses them as a testing ground for other results in the same paper).

In order to give the parameters of the codes, we use different methods. In the general case, we use a slight amelioration of the ideas present in \cite[Section 3.2]{hans}. We first ``cover'' the rational points of the surface by some numerically equivalent curves (this could -- very naturally -- be fibers, but we also work with images of sections). Then, in order to evaluate the number of rational points which are zeroes a given global section, we count the number of covering curves it contains and use intersection theory to bound the number of remaining rational points on the covering curves not contained in it. Finally, we consider a particular case, the \emph{unisecant} codes, for which all the possible decompositions in irreducible components of the global sections are easy to describe, as long as their numbers of rational points; this improves on the count based on the intersection number. 

Note a funny consequence of our constructions: when we consider particular unisecant codes, and use a bound from coding theory, we are able to give an upper bound on the arithmetic Segre invariant of a ruled surface in terms of the number of rational points of the base curve. This is a new hint that the link between the possible invariants of the ruled surfaces and the arithmetic properties of the base curve should be studied further.

We also construct asymptotically good sequences of codes. This gives an example of such codes coming from surfaces. This question has been raised in \cite{clp}; but note that our point of view is different, since we do not really construct towers of surfaces, rather ruled surfaces whose bases run over a tower of curves. Our main result here is that, starting with an asymptotically optimal sequence of base curves, it is possible to construct asymptotically good sequences of codes with better parameters that the ones obtained on the products of the base curves with the projective line.

Finally, the last two decades have seen a growing interest on the so-called \emph{local} properties of codes. We end the article with such considerations: for a given codeword, we show that its restrictions to the coordinates corresponding to the rational points of a fiber (\emph{resp.} a section) form a word of a Reed-Solomon code (\emph{resp.} of an AG code on the base curve). Then we show that evaluation codes on a ruled surface are, under a mild condition, locally recoverable, and we give results on their locality and availaibility.

\section*{Notations}

In the following, we consider a finite field $k:=\F_q$. We denote by $\overline{k}$ its algebraic closure, and by $\Gamma := \Gal(\overline{k}/k)$ its absolute Galois group, with generator $\tau : x\mapsto x^q$.

We denote by $C$ a projective nonsingular curve over $k$, of positive genus, by $K$ its function field, and by $N:=\#C(k)$ its number of rational points.

For $X$ a scheme over $k$, $\ell/k$ a field extension, we denote the extension of scalars by $X_\ell := X\times_{\Spec k} \Spec \ell$. Similarly, for $\varphi:X\rightarrow Y$ a morphism defined over $k$, we denote its extension of scalars by $\varphi_\ell : X_\ell \rightarrow Y_\ell$.

We assume familiarity with the construction of projective bundles from locally free sheaves as in \cite[(4.1.1)]{ega2}, \cite[(II.7)]{hart} and with the principal results on these objects (universal property, Picard group, canonical divisor, cohomology of invertible sheaves...).


\section{Ruled surfaces and projective bundles; first properties}

As we noted in the introduction, the definition of ruled surfaces has evolved over time; the one we give here is close to the modern definition over an the algebraically closed field \cite[(11.13)]{bade}, \cite[(V.2)]{hart}.

\begin{definition}
A nonsingular projective surface $X$ defined over $k$ is \emph{ruled over $C$} when there is a smooth morphism $\pi:X\rightarrow C$ defined over $k$ whose generic fiber is a smooth curve of genus $0$.

We call $C$ the \emph{base curve} of $X$.
\end{definition}

\begin{remark}
Such a surface is minimal, and even absolutely minimal (minimal over $\overline{k}$). Actually if $D$ is a rational curve on $X_{\overline{k}}$, then its image by $\pi_{\overline{k}}$ cannot be the whole $C_{\overline{k}}$ since it has positive genus, thus it is a point, and $D$ is contained in a fiber of $\pi_{\overline{k}}$. Since $\pi_{\overline{k}}$ is smooth, it must be the whole fiber. But the self-intersection of the fibers is $0$, and $D$ cannot be a $(-1)$-curve.
\end{remark}

A classical result \cite[(11.11)]{bade}, \cite[(V.2.2)]{hart}, \cite[\S 13 Theorem 4]{issh} tells us that the surface $X_{\overline{k}}$ is a projective bundle over $C_{\overline{k}}$: there exists some locally free sheaf of rank $2$ on $C_{\overline{k}}$, $\E$, such that we have an isomorphism $X_{\overline{k}}\simeq \P(\E)$ over $C_{\overline{k}}$. 

Let us recall briefly the arguments, since we shall use some of them below to prove that the same property holds over finite fields (see any of the references above, and in particular the proof of \cite[(11.10)]{bade} for more details). 

The generic fiber $X_{K\otimes \overline{k}}$ of $X_{\overline{k}}$ is a smooth curve of genus $0$ over the field $K\otimes \overline{k}$. We can use its anticanonical sheaf to embed it as a non singular conic in $\P^2_{K\otimes \overline{k}}$. Since $K\otimes \overline{k}$ is the function field of a curve over an algebraically closed field, it has property $C_1$ from Tsen's theorem, and the above conic has a rational point.

This rational point $\Spec K\otimes \overline{k}\rightarrow X_{K\otimes \overline{k}}$ gives us a rational map from $C_{\overline{k}}$ to $X_{\overline{k}}$, that is a section of $\pi_{\overline{k}}$. Since $C_{\overline{k}}$ is nonsingular and $X_{\overline{k}}$ projective, this map extends to a morphism. If we denote by $S$ its image, then the sheaf $\E:=\pi_{\overline{k}\ast}\O_{X_{\overline{k}}}(S)$ is a locally free sheaf of rank $2$ such that we have an isomorphism $X_{\overline{k}}\simeq \P(\E)$ over $C_{\overline{k}}$.

We shall extend this result to our situation. Almost all arguments above remain true, except Tsen's theorem. To show that the generic fiber $X_K$ has a rational point $\Spec K\rightarrow X_K$, we will instead apply the Hasse-Minkowski theorem to the quadratic form $Q$ defining the anticanonical embedding of the generic fiber $X_K$ in $\P^2_K$.

Let $v$ denote a place of $K$, and $K_v$ the corresponding completion. We denote by $\O_v$ its valuation ring, and by $k(v)$ its residue field, which is a finite extension of $k$. Now $X_{K_v}$ is also defined by the quadratic form $Q$ in $\P^2_{K_v}$, and we can assume that its coefficients are in $\O_v$. The reduction modulo $v$ of this conic is the fiber of $\pi$ over $\Spec k(v)$, and since $\pi$ is smooth, it is a smooth conic in $\P^2_{k(v)}$. Since $k(v)$ is finite, it has at least one (actually $\# k(v)+1$) rational point(s). Now Hensel's lemma ensures us that any such point lifts to a point with coefficients in $\O_v$. But this is a $K_v$-rational point of $X_{k_v}$ 

Finally, the quadratic form $Q$ over $K$ is isotropic over any completion $K_v$ of $K$, thus it is isotropic over $K$ by the Hasse-Minkowski theorem on quadratic forms over global fields. Reasoning as above, we deduce that $X_K$ has a $K$-rational point.

We have shown the first assertion of the following result.

\begin{theorem}
\label{theo1}
A ruled surface $\pi : X\rightarrow C$ over $k$ has the following properties

\begin{itemize}
\item[(1)] $\pi$ admits a section $\sigma$ which is defined over $k$; 
\end{itemize}
we fix such a section and denote its image by $S$;
\begin{itemize}
\item[(2)] $X$ is a projective bundle over $C$: there exists a rank $2$ locally free sheaf $\E$ on $C$ such that we have an isomorphism $X\simeq \P(\E)$ over $C$;
\item[(3)] let $f$ denote the class (for numerical equivalence) of a fiber of $\pi$. We have the following description of the divisor class groups of $X$
\[
\Pic X\simeq \Z S+\pi^\ast \Pic C,~\Num X \simeq \Z S+\Z f
\]
with intersection product satisfying $S^2=\deg\E$, $f^2=0$, $S\cdot f =1$;
\item[(4)] we have $\P(\E)\simeq \P(\E')$ over $C$ if, and only if there exists some $\L\in\Pic C$ such that $\E'\simeq \E\otimes \L$;
\item[(5)] the zeta function of $X$ is $Z(X,t)=Z(C,t)Z(C,qt)$;
\item[(6)] the linear equivalence class of the canonical divisor of $X$ is
\[
K_X\sim -2S+\pi^\ast (K_C+\det\E)
\] 
\end{itemize}

\end{theorem}

\begin{proof}
The proofs of the points (2) and (4) over an algebraically closed field remain valid (see \cite[(11.10) and (11.11)]{bade}). 

To prove the third one, remark that we have the following description of the geometric Picard group of a projective bundle
\[
\Pic \P(\E)_{\overline{k}} \simeq \Z \O_{\P(\E)}(1)+\pi^\ast \Pic C_{\overline{k}}
\]
Now from our construction of $\E\simeq \pi_\ast \O_X(S)$ we have $\O_X(S)\simeq \O_{\P(\E)}(1)$, and the isomorphism $\Pic X_{\overline{k}} \simeq \Z S+\pi^\ast \Pic C_{\overline{k}}$. As $k$ is finite and $X$ is smooth, we have $\Pic X=(\Pic X_{\overline{k}})^\Gamma$. Since $S$ is defined over $k$ the description of the Picard group follows from the equality $\Pic C=(\Pic C_{\overline{k}})^\Gamma$. 

The assertions on the intersection products $f^2$ and $f\cdot S$ follow from the facts that $f$ is a fiber of $\pi$ and $S$ the image of a section. To compute $S^2$, we use the asymptotic Riemann-Roch theorem \cite[(1.1.24)]{laza}: this self intersection number is twice the leading coefficient of the polynomial $n\mapsto \chi(nS)$. Since we have $S\cdot f=1$, we deduce from \cite[(V.2.4)]{hart} that we have $\chi(nS)=\chi(\pi_\ast \O_{\P(\E)}(n))$ (this last one is over $C$), and from \cite[(2.1.15)]{ega2} that  $\chi(\pi_\ast \O_{\P(\E)}(n))=\chi(\SS^n\E)$. This last sheaf is the $n$th symmetric product of $\E$; it has rank $n+1$, and degree $n(n+1)\deg \E/2$, thus we have $\chi(nS)=n(n+1)\deg \E/2+(n+1)(1-g)$, and $S^2=\deg \E$ as desired.

In order to show the fifth assertion, we first determine the cardinality of $X(\ell)$ for a finite extension $\ell/k$. If $x\in X(\ell)$ is such a point, then $\pi\circ x$ is a point in $C(\ell)$. Conversely, if we fix $p\in C(\ell)$, then any point $x\in C(\ell)$ with $\pi \circ x=p$ factors as a point in $\pi^{-1}(p)(\ell)$. Since this fiber is isomorphic to $\P^1_\ell$, we deduce that there are $\# \ell+1$ rational points in the fiber over $p$, and that we have $\# X(\ell)=(\#\ell +1)\# C(\ell)$. The formula for the zeta function follows by plugging these equalities in its exponential formula.

The assertion about the canonical divisor is a particular case of a general assertion about projective bundles. It can also be deduced from \cite[(V.2.10)]{hart}: there the author uses a \emph{normalized} sheaf $\E$, but replacing $\E$ by $\E'=\E\otimes\L$ gives $\O_{\P(\E')}(1)\simeq \O_{\P(\E)}(1)\otimes \pi^\ast \L$ and $\det\E'=(\det\E)\otimes \L^{\otimes 2}$. We see that the formula remains true for any rank two sheaf.
\end{proof}

\begin{corollary}
\label{extension}
Every locally free sheaf $\E$ of rank two on $C$ is an extension of invertible sheaves: there exists two elements $\M,\NN\in \Pic(C)$ such that the sequence
\[
0\rightarrow \M \rightarrow \E \rightarrow \NN \rightarrow 0
\]
is exact.
\end{corollary}

\begin{proof}
We apply the first assertion of the result above to the ruled surface $X=\P(\E)$ and we reason as in \cite[(V.2.6)]{hart}. Let $\sigma:C\rightarrow X$ be a section of $\pi$, with image $S$. The universal property of projective bundles \cite[(II.7.12)]{hart}, applied to this section, gives us a surjection $\E\rightarrow \sigma^\ast\O_{X}(1)$; its kernel is the invertible sheaf $\pi_{\ast}(\O_{X}(1)\otimes \O_{X}(-S))$.
\end{proof}

\section{Elementary transforms}

In this section, we generalize the notion of elementary transform to the non closed case. We will show in the next section that, as in the closed case, any ruled surface $X$ over $C$ can be obtained from the product $C\times \P^1$ by a composition of such maps. In applications, this will turn out to be a convenient way to construct new ruled surfaces from a given one. 

Let $X$ be a ruled surface over $C$, and $x$ a point of degree $d$ on $X$ (a Galois orbit of cardinality $d$ in $X(\overline{k})$), whose image $\pi(x)$ is a point of degree $d$ on $C$.

Let $\widetilde{X}$ denote the surface obtained by blowing up the point $x$ on $X$. 

We denote by $x_1,x_2,\ldots,x_d$ the points above $x$ in $X(\F_{q^d})$. Then $\widetilde{X}_{\overline{k}}$ is the surface obtained from  $X_{\overline{k}}$ by blowing up these $d$ points. It contains the (pairwise disjoint) exceptional divisors $E_1,\ldots,E_d$, which are $(-1)$-curves. 

Let $f_1,\ldots,f_d$ denote respectively the fibers of $\pi_{\overline{k}}$ in $X_{\overline{k}}$ over the pairwise distinct points $\pi_{\overline{k}}(x_1),\ldots,\pi_{\overline{k}}(x_d)$. Their strict transforms $\widetilde{f}_1,\ldots,\widetilde{f}_d$ are rational curves of self intersection $-1$, and they are also (pairwise disjoint) $(-1)$-curves. Moreover they are conjugate under the action of $\Gamma$.

As a consequence \cite[Theorem 3.2]{hass}, we can contract the divisor $E$ on $X$ corresponding to the union of these curves, and we obtain a new surface $\elm_x(X)$, where $\elm_x$ is the birational map over $k$ defined as the composition of the blowup at $x$ and the contraction of the fibers.

\begin{definition}
The birational map $\elm_x$ is the \emph{elementary transform on $X$ with center $x$}.
\end{definition}

The morphism $\pi : X\rightarrow C$ extends to a morphism $\widetilde{X}\rightarrow C$. Moreover, since the divisor we contract is contained in the fiber $\pi^{-1}(\pi(x))$, this morphism factors through $\pi_x : \elm_x(X)\rightarrow C$. The fibers of $\pi_x$ (including the generic one) are the same as that of $\pi$, except for the fiber over $x$, which is the image of $E$. Since this is a smooth curve of genus $0$, we get the first assertion of the following result.

\begin{proposition}
\label{elemtransfoppties}
Let $\pi:X \rightarrow C$ denote a ruled surface, and $x$ a point of degree $d$ on $X$, whose image $\pi(x)$ is a point of degree $d$ on $C$.

The image $\elm_x(X)$ of $X$ by the elementary transform with center $x$ has the following properties
\begin{itemize}
\item[(1)] it is a ruled surface over $C$; 
\item[(2)] if $S$ is the image of a section of $\pi$, and $E$ is the exceptional divisor of the blow up of $X$ at $x$, then we can identify the Picard group of $\elm_x(X)$ with the subgroup $\Z(S-E)+\pi^\ast \Pic C $ of $\Pic \widetilde{X}=\Pic X+\Z E$. The intersection product extends that of $\Pic X$ by $E^2=-d$ and $D\cdot E =0$ for any $D\in \Pic X$.
\item[(3)] if $D\sim aS+\pi^\ast \delta$, $\delta\in \Pic C$, is a curve on $X$, defined over $k$, then the class of the image of its strict transform in $\elm_x(X)$ is 
\[
\elm_x(D)\sim a(S-E)+\pi^\ast\left( \delta+(a-m_x(D))\pi(x)\right)
\] 
under the identification above, where $m_x(D)$ is the multiplicity of $D$ at $x$. In particular, the self intersection numbers satisfy
\[
\elm_x(D)^2=D^2+ad(a-2m_x(D))
\]
\item[(4)] the canonical divisor on $\elm_x(X)$ is linearly equivalent to $K_X+2E-\pi^\ast \pi(x)$.
\end{itemize}
\end{proposition}

\begin{proof}
Only the last three points remain to be proven. They are consequences of well known properties of blowups (see for instance \cite[(V.3)]{hart}, but remark that here we perform $2d$ monoidal transforms in the sense of that book: geometrically we blowup $d$ points, and then contract $d$ $(-1)$-curves).

The pullback of divisors by the blowup $\widetilde{X}_{\overline{k}}\rightarrow X_{\overline{k}}$ gives an (isometric for the intersection product) injection between the Picard groups, and we can identify the Picard group of the surface $\widetilde{X}_{\overline{k}}$ to the group generated by its image and the exceptional divisors $\Z S+\pi^\ast \Pic C_{\overline{k}}+\Z E_1+\ldots+\Z E_d$.

Since we contract the divisors $\widetilde{f}_i\sim \pi^\ast \pi(x_i)-E_i$, $1\leq i\leq d$, the Picard group of $\elm_x(X)_{\overline{k}}$ identifies to the orthogonal of these classes in $\Pic \widetilde{X}_{\overline{k}}$. We get the subgroup $\Z(S-E_1-\cdots-E_d)+\pi^\ast \Pic C_{\overline{k}}$, and the values for the intersection product come from the isometry above. Now $\{E_1,\ldots,E_d\}$ is an orbit for the Galois action on $\Pic \widetilde{X}_{\overline{k}}$, and we deduce the second point by taking the Galois invariants.

Let $D\sim aS+\pi^\ast \delta$ denote an irreducible curve on $X_{\overline{k}}$; it is well-known that the class of its strict transform in $\Pic \widetilde{X}_{\overline{k}}$ is $aS+\pi^\ast \delta-\sum m_{x_i}(D)E_i$ where $m_{x_i}(D)$ is the multiplicity of $D$ at $x_i$. If we assume that $D$ is defined over $k$, all multiplicities are equal and we denote their common value by $m_x(D)$. Thus the class of the strict transform of $D\sim aS+\pi^\ast \delta$ in $\Pic \widetilde{X}$ is $aS+\pi^\ast \delta-m_{x}(D)E$.

The corresponding class in $\Pic \elm_x(X)$ is the image of the above class by the orthogonal projection on the orthogonal of the class $\pi^\ast\pi(x)-E$, and we get the desired result for irreducible divisors. In general, it is sufficient to apply the above result to the irreducible components of an effective divisor.

The assertion about the canonical divisor is easy: the one of $\widetilde{X}$ is $K_X+E$, then since we contract a divisor linearly equivalent to $\pi^\ast \pi(x)-E$, we get the divisor $K_X+E-(\pi^\ast \pi(x)-E)$ as canonical divisor of $\elm_x(X)$. This is the desired result.
\end{proof}

\begin{remark}
Note, as a consequence of the construction, that if $y$ is the point of degree $d$ on $\elm_x(X)$ which is the image of the strict transform of the fiber above $x$, then the elementary transform $\elm_y$ on $\elm_x(X)$ is the inverse of the birational transform $\elm_x$.

\end{remark}

We now describe a rank $2$ locally free sheaf $\E'$ such that $\elm_x(X)=\P(\E')$, from a sheaf $\E$ such that $X=\P(\E)$. The new sheaf $\E'$ is usually called an elementary transform of $\E$, but we also find the terminology \emph{diminution \'el\'ementaire} in \cite{rayn}.

It is well known, when the base field is algebraically closed, that the new sheaf is the kernel of the surjection from $\E\rightarrow k(p)$ (the skyscraper sheaf concentrated at $p$) corresponding to the center of the elementary transform \cite[III. Exercice 2]{beau}, \cite[Section 7.4]{dolg}, \cite[Section 1]{hart3}, \cite[Proposition 1.1]{maru}.

This remains true in the non closed case, but we shall give a different description, better suited to our purposes.

\begin{proposition}
\label{elemtransfosheaf}
Notations are as above. Let $\E$ be a rank two locally free sheaf such that $X=\P(\E)$.

There exists a section $t : C\rightarrow X$  of $\pi$ such that $t(\pi(x))=x$, and if $T$ denotes its image, we have the exact sequence of sheaves 
 \[
0\rightarrow \pi_{\ast}(\O_{X}(1)\otimes \O_{X}(-T)) \rightarrow \E \rightarrow \L=t^\ast\O_{X}(1) \rightarrow 0
\]
Moreover, we have $\elm_x(X)=\P(\E')$, with $\E'=\E\times_{\L} \L(-\pi(x))$.
\end{proposition}

\begin{proof}
The images of the sections of $\pi$ are exactly the irreducible global sections of divisors of the form $D\sim S+\pi^\ast\L$, $\L \in \Pic C$. When $\deg\L$ is large enough, the linear system attached to such a class is large enough so that the subspace of those sections passing through $x$ (whose codimension is at most $d$) is non empty. 

Now from \cite[(2.6.1)]{ega3} the codimension of the space of those sections containing $\pi^\ast \pi(x)$ in the complete linear system above is $2d=\chi(D)-\chi(D-\pi^\ast \pi(x))$, again when $\deg\L$ is large enough. We deduce the existence of a section whose image passes through $x$ but does not contain $\pi^\ast \pi(x)$. It can contain other fibers of $\pi$, but once we have removed them, we end with an irreducible curve passing through $x$, that is the image of a section of $\pi$. The exact sequence describing $\E$ follows from the universal property of projective bundles applied to this section.

The point $x$ corresponds to a morphism $\Spec k(x)\rightarrow X$, and composing it with $\pi$ we get the point $\pi(x)$ on $C$. Applying once again the universal property of projective bundles, we get a surjection $\E\otimes k(\pi(x))\rightarrow k(\pi(x))$, and a surjection of $\O_C$-modules $\E \rightarrow k(\pi(x))$ (where this last field is the skyscraper sheaf at $\pi(x)$). We know from the references cited above that the kernel of this map is a sheaf $\E'$ such that $\elm_x(X)=\P(\E')$.

Since we have $s(\pi(x))=x$, the morphism $\Spec k(x)\rightarrow X$ factors through $\Spec k(\pi(x))\rightarrow C \rightarrow X$ where the second arrow is $s$. We deduce that the surjection $\E \rightarrow k(\pi(x))$ factors through the surjection  $\E \rightarrow \L$, and finally we deduce from the exact sequence 
\[
0\rightarrow \L(-\pi(x)) \rightarrow \L \rightarrow k(\pi(x)) \rightarrow 0
\]
the desired exact sequence
\[
0\rightarrow \E\times_{\L} \L(-\pi(x)) \rightarrow \E \rightarrow k(\pi(x)) \rightarrow 0
\]
\end{proof}

\section{Segre invariants}

Let $\E$ denote a rank $2$ locally free sheaf on $C$. The study of such objects has attracted much interest when $C$ is defined over an algebraically closed field. Our aim is to consider the generalisation of these results over finite fields. To the best of our knowledge, these questions have not been adressed in our setting, except incidentally in \cite{ball}.

The following lemma mimics \cite[Lemma 1.1]{maru}

\begin{lemma}
Let $\L\in \Pic(C)$ denote an invertible subsheaf of $\E$; then $\deg \L$ is bounded above.
\end{lemma}

\begin{proof}
We use Corollary \ref{extension} and the notations therein: if the composition of $\L\rightarrow \E$ with $\E\rightarrow \NN$ is non zero, then it must be an injection, and we have $\deg \L \leq \deg \NN$. Else the injection factors through $\L\rightarrow \M$, and we have $\deg \L \leq \deg \M$.
\end{proof}

We define two invariants associated to the sheaf $\E$; the first one has been thoroughly studied \cite{lana,maru}, but to our knowledge, there are almost no result about the second one, except \cite[Theorem 3]{ball} where it is called $1$-order of stability of $\E$. We decided to rename it in order to stress on its similarity with the Segre invariant of a locally free sheaf over an algebraically closed field.

\begin{definition}

Let $\E$ denote a rank $2$ locally free sheaf on $C$. We denote by $\overline{\E}$ the pull back of $\E$ under the extension of scalars $C_{\overline{k}}\rightarrow C$.

The \emph{geometric Segre invariant} of $\E$ is 
\[
s_g(\E):=\deg \E-2\max\left\{\deg\L,~\L\in \Pic(C_{\overline{k}}),~\L\hookrightarrow\overline{\E}\right\}
\]
The \emph{arithmetic Segre invariant} of $\E$ is 
\[
s_a(\E):=\deg \E-2\max\left\{\deg\L,~\L\in \Pic(C),~\L\hookrightarrow\E\right\}
\]
\end{definition}

\begin{remark}
Here are some immediate consequences of the definition:
\begin{itemize}
\item[(i)] we always have $s_g(\E)\leq s_a(\E)$;
\item[(ii)] we also have the congruences $s_g(\E)\equiv s_a(\E) \equiv \deg\E \mod 2$;
\item[(iii)] since $\L$ is a subsheaf of $\E$ if and only if $\L\otimes \M$ is a subsheaf of $\E\otimes \M$, the invariants above only depend on the ruled surface $\P(\E)$ from Theorem \ref{theo1} (4), and we will often speak about the Segre invariants of a surface.
\end{itemize} 
\end{remark}

The following geometric interpretation of these invariants is well-known in the closed case.

\begin{lemma}
\label{minself}
The arithmetic Segre invariant of $\E$ is the minimal self intersection number of the image a section of $\pi:\P(\E)\rightarrow C$ which is defined over $k$.

If we drop the assumption that the section must be defined over $k$, the minimal self intersection number above is the geometric Segre invariant.
\end{lemma}

\begin{proof}
Let $s$ denote a section of $\pi$ which is defined over $k$, with image $S$; it corresponds to a surjection $\E\rightarrow \L$, where $\L\in \Pic C$. From \cite[(V.2.9)]{hart}, we have $S\sim \O_{\P(\E)}(1)+\pi^\ast(\L-\det\E)$, and we deduce from Theorem \ref{theo1} (3) that its self intersection is $S^2=2\deg \L-\deg\E$.

The kernel $\M$ of the surjection $\E\rightarrow \L$ is an invertible subsheaf of $\E$ with $S^2=\deg \E-2\deg\M$, and this is the first part of the result.

The same proof, applied to the sections of $\pi_{\overline{k}} : \P(\overline{\E})\rightarrow C_{\overline{k}}$, gives the last assertion.
\end{proof}

\begin{remark}
\begin{itemize}
\item[(i)] the invariant $s_g(\E)$ is the opposite of the invariant $e$ defined in \cite[(V.2.8)]{hart}.
\item[(ii)] for any $\E$, we have the upper bound $s_g(\E)\leq g$ \cite{naga}.
\end{itemize}
\end{remark}

\begin{lemma}
\label{negseg}
Assume that we have $s_g(\E)<0$. Then we have $s_a(\E)=s_g(\E)$.

Moreover for any curve $C$ over $k$ and any integer $n> 0$ there exists a ruled surface $X\rightarrow C$ with $s_a(\E)=s_g(\E)=-n$.
\end{lemma}

\begin{proof}
As usual we consider the ruled surface $\pi:\P(\E)\rightarrow C$.

It follows from \cite[Lemma 1.2]{maru} that when $s_g(\E)<0$, there is a unique maximal invertible subsheaf of $\E$ in $\Pic C_{\overline{k}}$. Thus there is exactly one section of $\pi_{\overline{k}}$ whose image $S$ has self-intersection number $s_g(\E)$. Since the intersection product is invariant under the action of the Galois group $\Gamma$, we deduce that $S$ is defined over $k$, and $s_a(\E)=s_g(\E)$ from the preceding lemma.

To prove the second assertion, remark that for $d$ a large enough positive integer, the Weil bound ensures us that there exists two points $r_1$ and $r_2$ of respective degrees $d$ and $n+d$ on $C$. We consider the sheaf $\O_C(r_1)\oplus \O_C(r_2)$; its maximal subsheaf is $\O_C(r_2)$, and it is defined over $k$. Thus its Segre invariants are equal, and their value is $d+n+d-2(n+d)=-n$. The surface $X:=\P(\O_C(r_1)\oplus \O_C(r_2))$ satisfies $s_a(\E)=s_g(\E)=-n$.
\end{proof}

We are ready to prove that any ruled surface over $C$ can be obtained from the trivial one $C\times \P^1$ by a sequence of elementary transforms.

\begin{theorem}
\label{elemtransfoseq}
Let $X$ denote a ruled surface over $C$. Then there is a sequence of surfaces $X_i$, $1\leq i\leq n$, and points $x_i$ of degree $d_i$ on $X_i$, $1\leq i\leq n-1$ such that $X_1=C\times \P^1$, $X_n=X$, and $X_{i+1}=\elm_{x_i}(X_i)$ for any $1\leq i\leq n-1$.
\end{theorem}

\begin{proof}
We first show that this is true when (any of) the Segre invariant(s) of $X=\P(\E)$ is at most $1-2g$. In this case the exact sequence associated to a subsheaf of maximal degree splits \cite[(V.2.12)]{hart}, and we can write $X=\P(\E)$ with $\E\simeq \O_C\oplus \O_C(\delta)$. Since we have $\deg \delta \leq 1-2g$, $-\delta$ is an effective divisor, and we can write $\delta =-(p_1+\cdots+p_n)$ for some points (not necessarily distinct) $p_i$ of degree $d_i$ on $C$.

Now we start from $C\times \P^1=\P(\O_C\oplus\O_C)$, we consider the point $x_1:=(p_1,z)$ for some $z\in \P^1(k)$, and we perform the elementary transform with center $x_1$. The point $x_1$ is on the image of the section $p\mapsto (p,z)$ from $C$ to $C\times \P^1$, which corresponds to the surjective morphism from $\O_C\oplus\O_C$ to $\O_C$ defined by $(f,g)\mapsto z_0f+z_1g$ where $z=(z_0:z_1)$. We apply Proposition \ref{elemtransfosheaf} to this section and we get that $X_2:=\elm_{x_1}(C\times\P^1)\simeq \P(\E_2)$ with $\E_2=\O_C\oplus\O_C(-p_1)$.

On $X_2$, we consider $C_2$, the image of the strict transform of $C\times \{z\}$. This is the image of a section $s_2$, and we have $C_2\sim \O_{\P(\O_C\oplus\O_C)}(1)-E_1$ from Proposition \ref{elemtransfoppties} (3). Thus it has self-intersection $-d_1<0$, and it corresponds to the surjective morphism $\E_2\rightarrow \O_C(-p_1)$. Then we set $x_2:=s_2(p_2)$, and we consider the surface $X_3:=\elm_{x_2}(X_2)$.

Continuing in this fashion, we get the surface $X_{n+1}=\P(\E_{n+1})$ after $n$ steps, where $\E_{n+1}=\O_C\oplus \O_C(-p_1-\cdots-p_n)\simeq \E$. This is the desired surface.

We turn to the general case. Let $X$ denote a ruled surface over $C$, and $S$ the image of a section defined over $k$. Let us choose a point $x$ on $S$ of degree at least $S^2+2g-1$ (such a point exists by the Hasse-Weil lower bound). On the surface $\elm_x(X)$, the image of the strict transform of $S$ is the image of a section, and it has self intersection number at most $1-2g$ from Proposition \ref{elemtransfoppties} (3) (note that a section is smooth and consequently we have $m_x(S)=1$). Thus we are reduced to the preceding case.
\end{proof}

In general the inequality $s_a(X)\geq s_g(X)$ is strict; let us treat some examples.

Recall that the degree of a function on $C$ is the degree of its divisor of zeroes (or of poles). The degree of any constant function is zero.

\begin{lemma}
\label{segretrans}
Let $x$ denote a point of degree $e$ on $C\times \P^1$, such that $p_1(x)$ is a point of degree $e$ on $C$. Assume that $x$ does not lie on the graph of any function of degree at most $d$ in the algebraic function field $K$. Then the Segre invariant of the surface $X:=\elm_x(C\times \P^1)$ satisfies $s_a(X)\geq \min\{e,2(d+1)-e\}$.
\end{lemma}

\begin{proof}
From Lemma \ref{minself}, we have to give a lower bound for the possible self intersection numbers of the images of the sections of $\pi:X\rightarrow C$.

The image $S$ of such a section is the strict transform of the image $S_1$ of a section $s_1$ of the projection $p_1:C\times \P^1\rightarrow C$. Now $S_1$ is just the graph of the morphism $p_2\circ s_1$ from $C$ to $\P^1$, and it is numerically equivalent to $\O_{\P(\O_C\oplus\O_C)}(1)+d_1f$, where $d_1$ is the degree of the preceding morphism. 

Since $S_1$ is defined over $k$, it contains all the geometric points in the Galois orbit $x$ or it does not contain any. Now we use Proposition \ref{elemtransfoppties} (3). In the first case, the numerical class of the image in $X$ of its strict transform is $S\equiv \O_{\P(\O_C\oplus\O_C)}(1)-E+d_1f$ whose self intersection is $2d_1-e$. In the second case we get $S\equiv \O_{\P(\O_C\oplus\O_C)}(1)-E+(d_1+e)f$ for the class and $2d_1+e$ for the self intersection number. 

From our hypothesis, the first case can occur only if we have $d_1\geq d+1$, and this gives the desired lower bound on the self intersection numbers. 
\end{proof}

From this construction, we deduce that there exists ruled surfaces over a curve with high arithmetic Segre invariant, provided that its class number is not too large.

\begin{proposition}
\label{gplus1}
Let $C$ denote a curve of genus $g\geq 1$. Assume that its class number satisfies $h\leq q^g-q^{g-1}$, and that it has a point of degree $3g-1$. 

Then there is a ruled surface $X\rightarrow C$ with arithmetic Segre invariant $s_a(X)\geq g+1$.
\end{proposition}

\begin{proof}
Our strategy is to use the preceding lemma. We construct a point $x$ of degree $3g-1$ on $C\times \P^1$ which does not lie on the graph of any function of degree at most $2g-1$ in the function field $K$; then the elementary transform $X:=\elm_x(C\times \P^1)$ satisfies the desired property.

We count the functions of degree at most $2g-1$; if $f$ is such a function, we have $f\in L((f)_\infty)$: there exists some rational effective divisor $D\in \Eff_{2g-1}C$ such that $f\in L(D)$. Since any of the spaces $L(D)$ contains $k$, we can give an upper bound for the number $F_{2g-1}$ of functions of degree at most $2g-1$ in $K$ by considering the union $k\cup (\cup_D L(D)\setminus k)$ where $D$ runs in $\Eff_{2g-1}C$. We get
\[
F_{2g-1}\leq q + \sum_{D\in \Eff_{2g-1}C}(q^{h^0(D)}-q)=q + \sum_{D\in \Pic^{2g-1}C}|D|(q^{h^0(D)}-q)
\] 
where $|D|$ is the linear system associated to $D$, containing $(q^{h_0(D)}-1)/(q-1)$ divisors. From the Riemann-Roch theorem all divisors of degree $2g-1$ are non special, and we always have $h^0(D)=g$. Since $\Pic^{2g-1}C$ contains $h$ elements, we get the inequality $F_{2g-1}\leq q+h(q^{g}-q)(q^g-1)/(q-1)$.

If we have $h\leq q^g-q^{g-1}$, then we have $F_{2g-1}\leq q+q^{g-1}(q^g-1)(q^g-q)=q^{3g-1}-q^{2g}-q^{2g-1}+q^g+q<q^{3g-1}+1$. From our assumption there is a point $p$ of degree $3g-1$ on $C$, and if $r\in C(\F_{q^{3g-1}})$ is a point whose orbit under the action of $\Gamma$ is $p$, then its images by the functions of degree at most $2g-1$ cannot cover all the points in $\P^1(\F_{q^{3g-1}})$; thus there exists a point of degree $3g-1$ in $C\times \P^1$ which is not on the graph of any function of degree at most $2g-1$. 
\end{proof}

\begin{remark}
The consideration of other invariants of the curve may be relevant for the construction of ruled surfaces with high arithmetic Segre invariant. In the above proof, we have counted functions of bounded degree, and used this count in order to give an upper bound for the number of images of a fixed point of the curve. But the evaluation map need not be injective, and we may have fewer images. This happens for instance when the curve has a non trivial automorphism: in this case for any morphism $f:C\rightarrow \P^1$ and automorphism $\varphi$ of $C$, both $f$ and $f\circ \varphi$ share the same images.
\end{remark}

We now give an upper bound for the arithmetic Segre invariant. Note that it is attained in some cases (for instance when the base curve is elliptic, contains at most $q-1$ rational points and at least a point of degree two) from the proposition above.

\begin{proposition}
For any rank two sheaf $\E$ on $C$ of genus $g$, we have $s_a(\E)\leq 2g$.
\end{proposition}

\begin{proof}
We know from Theorem \ref{elemtransfoseq} that the surface $\pi:X=\P(\E)\rightarrow C$ is obtained from $C\times \P^1$ by a sequence of elementary transforms. We denote the consecutive surfaces by $X_i$ and the centers by $x_i$ (of degree $d_i$ on $X_i$), $1\leq i\leq n+1$. We have $X_1=C\times \P^1$, $X_{n+1}=X$, and $X_{i+1}=\elm_{x_i}(X_i)$ for any $1\leq i\leq n$. 

Let $d:=d_1+\ldots+d_n$, and $C_1\sim \O_{\P(\O_C\oplus\O_C)}(1)$ an ``horizontal'' section on $X_1$ (thus we have $C_1^2=0$).

Fix some $\L\in \Pic(C)$ of large enough degree; the dimension of the space of global sections of $C_1+\pi^\ast\L$ is $2h^0(\L)$, and the dimension of the subspace of those sections passing through $x_1$ is at least $2h^0(\L)-d_1$. This subspace becomes the complete linear system associated to $C_1-E_1+\pi^\ast\L$ on the blowup of $X_1$ at $x_1$, and since this last divisor is orthogonal to the divisor $\pi^\ast x_1-E_1$ that we contract in the second part of the elementary transform, it remains a divisor on $X_2$, with the same linear system.

We continue in the same fashion, and we obtain the linear system on $X_{n+1}=X$, associated to the divisor $C_1-E_1-\cdots-E_n+\pi^\ast\L$, of dimension at least $2h^0(\L)-d$. A section $S$ in this linear system has self intersection $2\deg\L-d$. If it is irreducible it is the image of a section of $\pi$; else there is a unique section of $\pi$ whose image is among the irreducible components of $S$, which must be defined over $k$, and has self-intersection less than $2\deg\L-d$.

We have constructed a section of $\pi$, defined over $k$, with self intersection $2\deg\L-d$ as long as we have $d<2h^0(\L)$. From Riemann-Roch theorem, it is sufficient to guarantee that we have $d<2(\deg\L+1-g)$, and we choose $d=2\deg\L+1-2g$ if $d$ is odd, and $d=2\deg\L-2g$ if $d$ is even. We get a section whose image has self intersection $2g-1$ in the first case, and $2g$ in the second.
\end{proof}

An interesting question is the following: given a genus $g$ curve over $k$ and a couple of integers $(s_g,s_a)$ satisfying the bounds $s_g\leq g$, $s_a\leq 2g$ and the congruence $s_g\equiv s_a \mod 2$, does there exist a ruled surface over $C$ whose Segre invariants are respectively $s_g$ and $s_a$? 

We have seen that when both numbers are negative, a necessary and sufficient condition is that they are equal. When they are positive, this condition no longer holds since we have constructed surfaces whose arithmetic Segre invariant is at least $g+1$, whence the geometric one is at most $g$ from Nagata's bound. We will see later that the answer to this question is negative; actually in Proposition \ref{segpts}, we prove that under a condition on the number of rational points of $C$, we can get strictly better bounds on the arithmetic invariants of sheaves on $C$.

This leads us to ask the following

\begin{question}
Given a genus $g$ curve defined over $\F_q$, for which couples of integers $(m,n)$ does there exist a ruled surface $\pi:X\rightarrow C$ defined over $\F_q$ such that $s_g=m$ and $s_a=n$ ?
\end{question}

Let us treat the case $g=1$, when $C$ is an elliptic curve. We have already constructed ruled surfaces with any negative invariant; of course the product $C\times \P^1$ gives the zero invariants; thus the only remaining couples are $(0,2)$ and $(1,1)$.

For the first couple, the two Segre invariants are distinct. Assume that $X$ is a surface with these invariants. We deduce that there are at least two sections whose images have minimal self-intersection, and that they are not defined over $k$. From \cite[Corollary 1.6 (i)]{maru}, we deduce that we have 
$X=\P(\E)$ for some $\E\simeq \O_C\oplus\O_C(D)$, $D$ a divisor of degree $0$. Since the curve is elliptic, we can write $D=p-q$ for two distinct points in $C(\overline{k})$. Since $X$ is defined over $k$, we must have $p,q\in C(\F_{q^2})$ and $p^\tau=q$. We deduce that we have $C(\F_q)\subsetneq C(\F_{q^2})$. 

Conversely, if this last condition is satisfied, then we choose some point $x$ on $C\times \P^1$ such that $p_1(x)\in C(\F_{q^2})\setminus C(\F_{q})$, and $p_2(x)\in \P^1(\F_{q^2})\setminus \P^1(\F_{q})$. The image of $C\times \P^1$ under the elementary transform with center $x$ has arithmetic Segre invariant $2$ from Lemma \ref{segretrans} since there is no degree one function on an elliptic curve.

We turn to the second couple $(1,1)$. We know from the Weil bound that $C$ has a $k$-rational point $p$, and we deduce from the Riemann-Roch theorem that the space $H^1(\O_C(-p))$ is one dimensional. Thus there exists a non trivial extension $\E$ of $\O_C(p)$ by $\O_C$. It has degree one, and geometric Segre invariant one since $\O_C$ is a maximal subsheaf. The only possible arithmetic Segre invariant is $1$, and we are done.


We summarize the results of the preceding discussion.

\begin{proposition}
Let $C$ denote an elliptic curve defined over $k=\F_q$. Then for any admissible couple $(s_g,s_a)\neq (0,2)$, there exists a ruled surface $X\rightarrow C$ having these Segre invariants. 

There exists a ruled surface over $C$ with invariants $(0,2)$ if and only if we have $C(\F_{q^2})\neq C(\F_q)$.
\end{proposition}

\section{Codes on ruled surfaces}

In this section, we fix a genus $g$ curve $C$ defined over $k$, and we consider certain ruled surfaces $\pi:X\rightarrow C$ over $k$; we fix an effective divisor $D$ and we denote by $\C$ the code obtained by evaluating the global sections of $D$ at the points of $X(\F_q)$.

\subsection{A first family}

Here we consider ruled surfaces which are the images of the product surface $C\times\P^1$ by some elementary transform.

Let $d\geq 2$ denote an integer. We assume that $C$ contains both a $k$-rational point and a point of degree $d$, and we choose a point $x$ of degree $d$ on $C\times \P^1$, such that its image on $C$ remains a point of degree $d$, and its image on $\P^1$ has degree at least $2$.

We denote by $X:=\elm_x(C\times\P^1)$ the ruled surface over $C$ defined as the image of the elementary transform with center $x$.

Recall from Proposition \ref{elemtransfoppties} that if we denote by $C_0$ the class of the invertible sheaf $\O_{\P(\O_C\oplus\O_C)}(1)$ in $\Pic (C\times\P^1)$, we can identify the group $\Pic X$ with $\Z(C_0-E)+\pi^\ast\Pic C$.

We fix a divisor $D$ in the numerical class $a(C_0-E)+bf$, where $a,b\geq 0$, and we consider the code obtained by evaluating the global sections of $D$ at the points of $X(\F_q)$.

It follows from Theorem \ref{theo1} (5) that the length of the code is $(q+1)N$ (this is why we assumed this set is not empty).

To determine its dimension, we use the Riemann-Roch theorem on $X$. From the description of the canonical divisor on $X$ in Proposition \ref{elemtransfoppties}, we get the Euler characteristic
\[
\chi(\O_X(D))=\frac{1}{2}D\cdot(D-K_X)+\chi(\O_X)=(a+1)(b+1-g)-d\frac{a(a+1)}{2}
\]
Since we have $a\geq 0$, we deduce from \cite[(V.2.4)]{hart} that we have $h^2(\O_X(D))=h^2(\pi_\ast\O_X(D))$. But this last dimension vanishes: we have a cohomology space coming from Zariski's topology on a curve. As a consequence, we have $\chi(\O_X(D))=h^0(\O_X(D))-h^1(\O_X(D))$, and $\chi(\O_X(D))$ is a lower bound for the dimension of the space of global sections of $D$. We also get a lower bound for the code dimension, as long as the evaluation map is injective, ie when the minimum distance is positive.

In order to estimate the minimum distance, we use some ideas that date back to Hansen \cite[Section 3.2]{hans}. 

We first need some irreducible curves on $X$ that cover all the evaluation points. We choose the images $F_1,\ldots,F_{q+1}$ by the elementary transform of the fibers of $p_2:C\times\P^1\rightarrow\P^1$ over the points in $\P^1(k)$. None of these fibers passes through $x$ (we assumed that its image on $\P^1$ has degree at least $2$), and since their class is $C_0$ in $\Pic (C\times\P^1)$, we deduce from Proposition \ref{elemtransfoppties} (3) that we have $F_i\equiv C_0-E+df$.

Let $s$ denote a global section of $\O_X(D)$, and $S$ the corresponding Weil divisor on $X$. We wish to bound from above the number of rational points $\#S(k)$. We slightly refine Hansen's result by taking into account the number of covering curves contained in a section all along the calculation.

Assume that, among its irreducible components, $S$ contains exactly $n$ curves among the $F_i$ above. The union $S'$ of the remaining irreducible components of $X$ is numerically equivalent to $(a-n)(C_0-E)+(b-dn)f$. As a consequence, the number of intersection points of $S'$ with any of the $q+1-n$ remaining curves $F_i$ is at most the intersection number $S'\cdot F_i=b-dn$, and since all rational points on $X$ lie on exactly one of the curves $F_i$, we get the bound
\[
\# S(k)\leq nN+(q+1-n)(b-dn)
\]
Now we let $n$ vary. The above expression is a convex function of $n$, and it is maximal at the extremities of the interval containing the possible values for $n$. 

We determine this interval: any effective divisor on $X$ which is not numerically equivalent to $f$ is the image of an effective divisor on $C\times \P^1$ by the elementary transform. If a divisor numerically equivalent to $a'C_0+b'f$ on $C\times \P^1$ is effective, then we have $a',b'\geq 0$, and the numerical class of the image in $X$ of its strict transform is $a'(C_0-E)+(b'+d(a'-m'))f$, where $m'$ is its multiplicity at $x$, which is at most $a'$. We deduce that if a divisor on $X$ is numerically equivalent to $s(C_0-E)+tf$, then we have $s,t\geq 0$.

If we apply this to the divisor $S'$ above, we get $a-n\geq 0$ and $b-dn\geq 0$, that is $0\leq n\leq \min\{a,\lfloor\frac{b}{d}\rfloor\}$. Finally we have proven the 

\begin{proposition}
\label{famcodes1}
Let $X$ denote the image of $C\times\P^1$ by the elementary transform with center a point $x$ of degree $d$ as above.

We fix any divisor $D\equiv a(C_0-E)+bf$ on $X$, and set $m:= \min\{a,\lfloor\frac{b}{d}\rfloor\}$.

Assume that the integers $a$ and $b$ satisfy the inequalities $0\leq b<N$, $0\leq m<q+1$ and $ad<2(b+1-g)$. The code obtained by evaluating the global sections of $D$ at the rational points of $X$ has the following parameters
\begin{itemize}
 \item[(i)] its length is $(q+1)N$;
  \item[(ii)] its dimension is at least $(a+1)(b+1-g)-d\frac{a(a+1)}{2}$;
  \item[(iii)] its minimum distance satisfies 
\[
d_{\rm min}\geq \min\{(q+1)(N-b),~(q+1-m)(N-b+dm)\}
\]
 \end{itemize} 
\end{proposition}

\subsection{A second family}

In this section, we construct codes on another type of ruled surfaces, with negative invariants, that we have already considered in the proof of Lemma \ref{negseg}.

We choose a divisor $\delta$ on $C$ of positive degree, and we consider the surface $\pi : X\rightarrow C$, where $X:=\P(\O_C\oplus\O_C(-\delta))$. The surjection $\O_C\oplus\O_C(-\delta)\rightarrow \O_C(-\delta)$ induces a section of $\pi$ whose image we denote by $S$. This is the section with minimal self intersection number $-\deg\delta$, and we deduce from Lemma \ref{negseg} again that the Segre invariants are equal to $s_a(X)=s_g(X)=-\deg \delta$.

\begin{remark}
Note that this surface can easily be constructed as the image of the product surface by a sequence of elementary transforms.

Write the divisor $\delta=\delta_1-\delta_2$ as a difference of effective divisors. If we apply to $C\times\P^1$ and $\delta_2$ the procedure described at the beginning of the proof of Theorem \ref{elemtransfoseq} for some $z\in\P^1(k)$, we get the projective bundle associated to $\O_C\oplus\O_C(-\delta_2)$.

On this projective bundle, we consider the section associated to the surjection $\O_C\oplus\O_C(-\delta_2)\rightarrow \O_C$, and we perform the elementary transforms whose centers are the points of this section and its strict transforms, associated to the points in the support of $\delta_1$. In this way we get the projective bundle associated to $\O_C(-\delta_1)\oplus\O_C(-\delta_2)$ from Proposition \ref{elemtransfosheaf}.

From Theorem \ref{theo1} (4), this is the desired ruled surface over $C$.
\end{remark}

We follow the notations in \cite{hart} here, and set $e=\deg\delta$.

We choose the divisor $D$ in the numerical class $aS+bf$, and we consider the code $\C$ obtained by evaluating of the global sections of $D$ at the rational points of $X$.

Let us look at its parameters; the length remains unchanged from the above construction. To compute the dimension, we reason as above, with the help of Theorem \ref{theo1}. We get the Euler characteristic 
\[
\chi(\O_X(D))=(a+1)(b+1-g)-e\frac{a(a+1)}{2}
\]
which gives, reasoning the same way as in the preceding section, a lower bound for the dimension of the code when the evaluation map is injective.

To compute the minimum distance, we use Hansen's results once again. Here we choose the curves $F_1,\ldots,F_N$ to be the fibers of $\pi$ over the $N$ rational points of $C$. Note that each one is isomorphic to $\P^1$ and contains $q+1$ rational points.

We begin with a lemma that ensures us that in some cases, all sections of a given divisor contain the negative curve $S$ with some multiplicity. It will be useful to give a bound (better than the one we would obtain via a direct application of Hansen's results) on the maximal number of rational points of such sections.

\begin{lemma}
Let $D\equiv uS+vf$ be an effective divisor on $X$. Then we must have $v\geq 0$.

Assume that we have $v<ue$; if we set $i:=u-\lfloor\frac{v}{e}\rfloor$, then for any global section $T$ of $D$, we have $T\geq iS$.
\end{lemma}

\begin{proof}
The first assertion comes directly from \cite[(V.2.20)]{hart}.

To show the last one, recall the following fact: if on a surface we have two divisors $C$ an irreducible curve and $D$ effective such that $C\cdot D<0$ then $C$ is a fixed component of the linear system $|D|$.

Since we assumed $v<ue$, we have $D\cdot S=v-ue<0$, and $S$ is a fixed component of the linear system $|D|$. We continue in the same fashion replacing $D$ by $D-S$ until the intersection product becomes nonnegative. We conlude that $S$ is a fixed component of the linear system $|D|$ with multiplicity $i$. This is the desired result.
\end{proof}

We are ready to analyse the geometry of the global sections of the divisor $D$. We will begin by fixing the number of fibers $F_i$ such a section contains, and deduce the multiplicity of the negative curve $S$ in this section accordingly. This gives an upper bound on the number of points depending on the above number of fibers.

Let $T\in |D|$, seen as a Weil divisor. Assume that it contains exactly $t$ of the fibers $F_i$. We have $T=U+F_{i_1}+\cdots+F_{i_t}$ where $U$ is an effective divisor, numerically equivalent to $aC_0+(b-t)f$. Note that we must have $b-t\geq 0$.

We apply the above lemma to the divisor $U$: 
\begin{itemize}
	\item[(i)] if $b-t < ae$, we can write $U=iS+V$ with $i=a-\left\lfloor\frac{b-t}{e}\right\rfloor$, and $V$ is a global section of (some divisor numerically equivalent to) $(a-i)C_0+(b-t)f=\left\lfloor\frac{b-t}{e}\right\rfloor C_0+(b-t)f$. The section $V$ meets each of the $N-t$ remaining fibers $F_i$ in at most $\left\lfloor\frac{b-t}{e}\right\rfloor$ points since we have $(\left\lfloor\frac{b-t}{e}\right\rfloor C_0+(b-t)f)\cdot f=\left\lfloor\frac{b-t}{e}\right\rfloor$. Finally, the number of rational points of our section $T$ is at most 
\[
\# T(\F_q)\leq (q+1)t+(N-t)+\left\lfloor\frac{b-t}{e}\right\rfloor(N-t)
\]
where the first term comes from the $t$ fibers contained in the support of $T$, the second from the intersection points of the curve $S$ with the $N-t$ remaining fibers, and the last one from the intersections of $V$ with the remaining fibers.

\item[(ii)] if $b-t \geq ae$, the section $T$ meets each of the $N-t$ remaining fibers in at most $a$ points, and we get the upper bound
\[
\# T(\F_q)\leq (q+1)t+a(N-t)
\]
\end{itemize}

Let us summarize these results: the number of rational points of a global section $T\in|D|$ containing $t$ fibers satisfies

\begin{equation}
\label{ineq}
\# T(\F_q)  \leq   \left\{
\begin{array}{rcl}
qt+N+\left\lfloor\frac{b-t}{e}\right\rfloor(N-t) & \textrm{if} & t>b-ae \\
(q+1-a)t+aN & \textrm{if} & t\leq b-ae \\
\end{array}
\right.
\end{equation}

We want to give an upper bound for these numbers when $t$ varies. We treat separately two cases, depending on the sign of $b-ae$.

First if we have $b<ae$, the second case in (\ref{ineq}) cannot appear, and the maximal number of points of a section is given, when $t$ varies, by the piecewise affine function
\[
t\mapsto
\left\{
\begin{array}{rcl}
(q-\left\lfloor\frac{b}{e}\right\rfloor)t+(\left\lfloor\frac{b}{e}\right\rfloor+1)N & \textrm{over} & [0,b-\left\lfloor\frac{b}{e}\right\rfloor e]\\
(q-i)t+(i+1)N & \textrm{over} & ]b-(i+1)e,b-ie],~0\leq i\leq \left\lfloor\frac{b}{e}\right\rfloor-1\\
\end{array}
\right.
\]
We assume the inequality $b<qe$ in the following: else the above function reaches the value $(q+1)N$ when $i=q$, and the evaluation map is no longer injective.

In this case each piece has positive slope, and we just have to consider the values at the right of each interval (when $t=b-ie$). These values are a convex quadratic function of $i$, and their maximum is reached for $i=0$ or $\left\lfloor\frac{b}{e}\right\rfloor$; we get 
\[
\# T(\F_q)\leq \max\left\{\left(q-\left\lfloor\frac{b}{e}\right\rfloor\right)\left(b-\left\lfloor\frac{b}{e}\right\rfloor e\right)+\left(\left\lfloor\frac{b}{e}\right\rfloor+1\right)N,qb+N\right\} 
\]

In the second case $b\geq ae$, we assume $a\leq q$ else the bound in the second line of (\ref{ineq}) reaches the value $(q+1)N$ and the evaluation map is not injective. 

We get the following upper bound on the number of rational points of a global section as a function of the number $0\leq t\leq b$ of fibers it contains
\[
t\mapsto\left\{
\begin{array}{rcl}
(q+1-a)t+aN & \textrm{over} & [0,b-ae]\\
(q-i)t+(i+1)N & \textrm{over} & ]b-(i+1)e,b-ie],~0\leq i\leq a-1\\
\end{array}
\right.
\]
If we have $a=0$, we get at most $(q+1)b$ points. Assume $a$ is positive; reasoning as above, we see that we have to consider the values at $b-ae$ and the $b-ie$, $0\leq i\leq a-1$ in order to get the maximum. The last ones are the values of a (quadratic) convex function of $i$, and their maximum is attained for $i=0$ or $i=a-1$. Since the value for $i=a-1$ (or $t=b-ae+e$) is greater than the one for $t=b-ae$, we get
\[
\# T(\F_q)\leq \max\left\{(q+1-a)(b-ae+e)+aN,qb+N\right\} 
\]

Summarizing, we get the parameters of the codes in our second family

\begin{proposition}
\label{famcodes2}
Let $\pi:X\rightarrow C$ denote a ruled surface with Segre invariants $-e<0$. Let $S$ denote the negative section of $\pi$.

Assume that we have $0\leq a\leq q$, $0\leq b<N$ and $ae<2(b+1-g)$; then the parameters of the evaluation code associated to a divisor in the numerical class $aS+bf$ satisfy
\begin{itemize}
	\item[(i)] its length is $n=(q+1)N$;
		\item[(ii)] its dimension is at least $ (a+1)(b+1-g)-e\frac{a(a+1)}{2}$;
			\item[(iii)] its minimum distance is at least
	\[
	\left\{
\begin{array}{rcl}
\min\left((q+1-\left\lfloor\frac{b}{e}\right\rfloor)(N-b+\left\lfloor\frac{b}{e}\right\rfloor e),q(N-b)\right) & \textrm{if} & b<ae\\
\min\left((q+1-a)(N-b+(a-1)e),q(N-b)\right) & \textrm{if} & b\geq ae,~a\neq 0\\
(q+1)(N-b) & \textrm{if} & a= 0\\
\end{array}
\right.
	\].
\end{itemize}
\end{proposition}

\subsection{Unisecant codes}

We do not make any hypothesis on the ruled surface $\P(\E)$ in this section. We consider the case $a=1$, which allows us to describe precisely the possible decompositions into irreducible components of the global sections of $D$. As a consequence, we get better bounds for the parameters of these particular codes than with Hansen's method.

\begin{definition}
The code $\C$ obtained by evaluating the global sections of $D\equiv a\O_{\P(\E)}(1)+bf$ at the points of $\P(\E)(\F_q)$ is a \emph{unisecant code} when we have $a=1$.
\end{definition}

The parameters of such a code are rather easy to describe

\begin{proposition}
\label{unicode}
Let $\L\in\Pic C$ denote an invertible sheaf; we consider the code defined by the evaluation of the global sections of a divisor $D\sim\O_{\P(\E)}(1)+\pi^\ast \L$ at the points of $\P(\E)(k)$.

Recall that $s_a(\E)$ is the arithmetic Segre invariant of $\E$. The parameters $[n,k,d]$ of this code satisfy $n=(q+1)N$ and
\[
k\geq \deg\E+2(\deg \L+1-g),~d\geq q(N-(\deg\E-s_a(\E))/2-\deg\L)
\]
when both right hand sides of the inequalities are positive.
\end{proposition}

\begin{proof}
The length comes from Theorem \ref{theo1} (5). To estimate the dimension, we first use the Riemann-Roch theorem on $\P(\E)$. From the description of the canonical divisor on $\P(\E)$ in the same theorem, we get the Euler characteristic of $D$ (note that we have $\chi(\O_X)=1-g$)
\[
\chi(D)=\frac{1}{2}D\cdot(D-K_{\P(\E)})+\chi(\O_X)=\deg\E+2\deg \L+2-2g
\]
Reasoning as in the preceding cases, the dimension of $H^0(X,D)$ is at least the above number. This will give us a lower bound for the dimension of the code if we show that the evaluation map is injective.

We turn to the minimum distance. We use a different strategy here: for a given global section, we determine its possible decompositions in irreducible components. These ones are rather simple since they must be fibers, or images of sections of $\pi$; as a consequence we can give a bound for their numbers of rational points, which is exactly what we need for the determination of the minimum distance.

Consider a global section of $D$. If it is irreducible, the restriction to it of the morphism $\pi$ has degree one, and it is isomorphic to the curve $C$: it contains $N$ points. Else from \cite[V.2.20$\&$21]{hart}, we write it as the union of an irreducible curve $C'\sim\O_{\P(\E)}(1)+\pi^\ast \L'$ (which must be isomorphic to $C$ for the same reason) and $u$ fibers: then it contains at most $N+qu$ points (when all the fibers are above rational points of $C$). The numerical equivalence class of $C'$ is $\O_{\P(\E)}(1)+(\deg \L-u)f$, and its self intersection is $2(\deg\L-u)+\deg\E$. Since this number is at least $s_a(\E)$, we get the inequality $u\leq (\deg\E-s_a(\E))/2+\deg\L$, and finally the minimum distance is at least $q(N-(\deg\E-s_a(\E))/2-\deg\L)$.
\end{proof}

\begin{remark}
One can verify that the bounds on the parameters given in the preceding result are better than the ones we would obtain by using Hansen's ideas. The reason is that in the particular case of unisecant codes, we can describe precisely the possible decompositions of the global sections into irreducible components. 

It would be desirable to extend this to more general codes on a ruled surface $X$; but this would require determining the numerical classes of irreducible curves on $X$, and this is already a difficult question over an algebraically closed field of positive characteristic (see for instance \cite[Exercises V.2.14 and 15]{hart}).
\end{remark}

We deduce from this result a funny application of a result from coding theory, the Griesmer bound. It says that the arithmetic Segre invariant has something to do with the number of rational points: when a curve has many rational points, the arithmetic Segre invariants of the locally free sheaves of rank two on $C$ cannot be too high.

\begin{proposition}
\label{segpts}
Let $t\geq 0$ denote an integer. Assume that the number of rational points of the base curve satisfies 
\[
N> \max\{(t+1)(q^2+1),t(q^2+q+1)\}
\]
then for any locally free sheaf $\E$ of rank two on $C$, we have $s_a(\E)<2(g-t)$.
\end{proposition}

\begin{proof}
Assume that the inequality on the number of points is satisfied, and that there exists some locally free sheaf $\E$ of rank two on $C$ with arithmetic Segre invariant $s_a(\E)\geq 2g-2t$.

We first assume that $\deg\E$ is even, and choose some $\L\in \Pic(C)$ with $\deg\E+2\deg\L=2g+2$. Then the minimum distance of the code constructed in Proposition \ref{unicode} is at least $q(N-t-1)$, and its dimension at least $4$ since $C$ contains more than $t+2$ rational points. Thus we deduce the following inequality from the Griesmer bound
\[
(q+1)N\geq q(N-t-1)+N-t-1+\left\lceil\frac{N-t-1}{q}\right\rceil+\left\lceil\frac{N-t-1}{q^2}\right\rceil
\]
We can remove the ceilings, and we get the inequality $N\leq (t+1)(q^2+1)$ that contradicts our assumption on the number of points.

When $\deg\E$ is odd, its arithmetic Segre invariant is at least $2g-2t+1$ from our hypothesis, and we choose some $\L\in \Pic(C)$ with $\deg\E+2\deg\L=2g+1$; then we reason the same way as in the even case. This time we get the inequality $N\leq t(q^2+q+1)$ from the Griesmer bound, violating once again our assumption.
\end{proof}

\section{Asymptotically good families}


Our aim here is to construct asymptotically good families of codes on ruled surfaces, whose parameters are better than the ones of the codes on the product $C\times \P^1$. 

We first consider evaluation codes on the product, and estimate their parameters. This could be done by setting $e=0$ in the second family of codes constructed above, but it is easy to give an alternative and more precise description.

We denote by $Y=\P(\O_C\oplus\O_C)=C\times \P^1$ the product ruled surface. Recall that we have denoted the two projections of $Y$ on its factors by $p_1$ and $p_2$ respectively. Since we have $\O_Y(a)=p_2^\ast\O_{\P^1}(a)$ for any integer $a$, the following result is an immediate consequence of the Künneth formula for the global sections of a sheaf on a product of varieties. In the case $C=\P^1$, this has already been observed \cite[Theorem 2.1]{codu}.

\begin{proposition}
Let $D=\O_Y(a)+p_1^\ast \beta$ denote a divisor on $Y$ with $\beta\in \Pic C$ and $a,\deg \beta\geq 0$. 

We denote by $\PRS(a)$ the projective Reed-Solomon code (also called doubly extended) on $\P^1$ associated to homogeneous polynomials of degree $a$, and by $\C_C(\beta)$ the code obtained by evaluating the global sections of $\beta$ at the points of $C(k)$.

Then the code $\C_Y(D)$ obtained by evaluating the global sections of $D$ at the points of $Y(k)$ is the tensor product $\PRS(a)\otimes \C_C(\beta)$. 
\end{proposition}

It is well-known that the parameters of the Reed-Solomon code are $[q+1,a+1,q+1-a]$, and the ones of the code on the curve are $[N,\geq b+1-g,N-b]$. As a consequence the paramaters of the code $\C_Y(D)$ are bounded below by the products $[(q+1)N,\geq (a+1)(b+1-g),(q+1-a)(N-b)]$.

We now consider sequences of codes whose length grows to infinity. The following definition will be useful when we want to compare their parameters.

\begin{definition}
If $\C$ is an $[n,k,d]$-linear code, we define its \emph{rate} by $R:=\frac{k}{n}$, and its \emph{relative distance} by $\delta:=\frac{d}{n}$.

A sequence $(\C_i)_{i\geq 0}$ of linear codes, with parameters $[n_i,k_i,d_i]$ is \emph{asymptotically good} when we have $\lim n_i=+\infty$, and the limits of the rates $\liminf R_i$, $\liminf\delta_i$ are positive. 
\end{definition}

Recall that $k=\F_q$. We fix once and for all an asymptotically optimal sequence of curves $(C_i)_{i\geq 0}$; recall that if $C_i$ has genus $g_i\rightarrow \infty$ and $N_i=\#C_i(k)$ rational points, this means that the quotients $N_i/g_i$ tend to $A(q)$ which is the highest possible limit for such a sequence.

We first consider a sequence of codes $\C_i$ on the products $C_i\times \P^1$. We fix some $1\leq a\leq q$ and for each $i\geq 0$ some divisor $\beta_i$ on $C_i$ with degree $b_i$. Then we consider the code $\C_i$ obtained by evaluating the global sections of the divisor $\O_{\P(\O_C\oplus\O_C)}(a)+p_1^\ast \beta_i$ at the points of $C_i\times \P^1(k)$. From the above proposition, its relative parameters are the products of the relative parameters of the Reed-Solomon code (which is MDS) and of the code on the curve. We get 
\[
R_i=\frac{a+1}{q+1}\frac{b_i+1-g_i}{N_i},~\delta_i=\left(1-\frac{a}{q+1}\right)\left(1-\frac{b_i}{N_i}\right)
\]
We let $i\rightarrow\infty$, and we assume that the quotients $\frac{b_i}{N_i}$ converge to a limit $\frac{1}{A(q)}<b<1$. Then the relative parameters converge to a limit $(\delta,R)$ with $\delta$ and $R$ positive, and lying on the line
\[
tx+   \left(1+\frac{1}{q+1}-t\right) y  =t\left(1+\frac{1}{q+1}-t\right)\left(1-\frac{1}{A(q)}\right),~t=\frac{a+1}{q+1}
\]  
Thus the relative parameters of the best product codes lie on or above these $q$ lines. In order to have a clearer view, we consider the parameter $t$ as a real number varying in the interval $[0,1]$. Then we replace the familiy of lines above by their envelope, and we conclude that the relative parameters of the best product codes lie on or above the curve parametrised by 
\[
x(t)=Bt^2,~y(t)=B(1-t)^2,~0\leq t \leq 1,~B=\left(1-\frac{1}{A(q)}\right)\left(1+\frac{1}{q+1}\right)
\]
which is the blue curve on the diagrams below.

We now construct asymptotic families of codes on ruled surfaces. We start from a sequence of curves $(C_i)$ as above, and for each one we define
\begin{itemize}
\item[(i)] a point $x_i$ of degree $d_i$ on $C_i\times \P^1$, whose image on $C_i$ has degree $d_i$, and on $\P^1$ has degree $>1$;
\item[(ii)] the ruled surface $\pi_i:X_i:=\elm_{x_i}(C_i\times\P^1)\rightarrow C_i$ obtained as the image of the elementary transform of $C_i\times \P^1$ with center $x_i$;
\item[(iii)] a divisor $D_i\sim a_i(\O_{\P(\O_C\oplus\O_C)}(1)-E_i)+\pi_i^\ast \beta_i$ on $X_i$, with $1\leq a_i\leq q$ and $\deg \beta_i=b_i\leq N_i$.
\end{itemize}

Then the code $\C_i$ is the code obtained by the evaluation of the global sections of the divisor $D_i$ at the rational points of the surface $X_i$.

It follows from Proposition \ref{famcodes1} that if we set $m_i=\min\{a_i,\lfloor\frac{b_i}{d_i}\rfloor\}$, the relative parameters of the code $\C_i$ are 
\begin{eqnarray*}
\delta_i & \geq & \min\left\{1-\frac{b_i}{N_i},\left(1-\frac{m_i}{q+1}\right)\left(1-\frac{b_i}{N_i}+\frac{d_i}{N_i}m_i\right)\right\},\\
R_i & \geq & \frac{a_i+1}{q+1}\frac{b_i+1-g_i-\frac{1}{2}a_id_i}{N_i}
\end{eqnarray*}

\begin{remark}
We consider codes from the first family since their parameters are slightly better than the ones from the second family.
\end{remark}

We assume in the following that the sequences $(\frac{b_i}{N_i})$ and $(\frac{d_i}{N_i})$ converge respectively to $b$ and $d$. We also assume that the limit $a=\lim \frac{a_i}{q+1}$ exists, and we treat it as a real number in the interval $[0,1]$. Recall that the sequence $(\frac{N_i}{g_i})$ converges to $A(q)$ from our hypothesis that the family of base curves is asymptotically optimal. Note that the sequence $(\frac{m_i}{q+1})$ tends to $m=\min\{a,\frac{1}{q+1}\lfloor\frac{b}{d}\rfloor\}$.

The sequences of relative parameters $\delta_i,R_i$ of the codes $\C_i$ tend respectively to the values $\min\left\{1-b,(1-m)(1-b+(q+1)md)\right\}$ and
\[
 \left(a+\frac{1}{q+1}\right)\left(b-\frac{1}{A(q)}-\frac{1}{2}(q+1)ad\right)
\]

We restrict our attention to the domain where $m=a$ and $1-b\leq (1-m)(1-b+(q+1)md)$. On one hand, this will be sufficient to construct codes asymptotically better then the product codes, and on the other hand numerical experiments show that the best relative parameters are attained in this domain.

We verify that $m=a$ when $d\leq\frac{b}{(q+1)a}$, and in this case the inequality $1-b\leq (1-m)(1-b+(q+1)md)$ is equivalent to $d\geq\frac{1-b}{(q+1)(1-a)}$. Note that we must have $a\leq b$ for the interval containing $d$ to be non empty. 

Under the above hypotheses, the relative minimum distance tends to $1-b$. Thus we fix $b$ and try to get the highest value for the limit of the rates given above. This is a decreasing function of the variable $d$, and we set $d=\frac{1-b}{(q+1)(1-a)}$. Now the function $R$ is maximal for the following value of $a$
\[
a_0=1-\sqrt{\frac{(q+2)A(q)(1-b)}{(q+1)(A(q)(b+1)-2)}}
\]
One checks easily that this is well-defined as soon as $A(q)> 2$ and that we have $a_0\leq b$ as necessary. Now the maximal rate is given by 
\[
1-\frac{1}{A(q)} +\frac{A(q)(b+1)-2}{2(q+1)A(q)}-\sqrt{\frac{(q+2)A(q)(1-b)}{(q+1)(A(q)(b+1)-2)}}
\]
Making $b$ vary gives us a new curve in the $(\delta,R)$ plane, that we draw in red in the figures below.

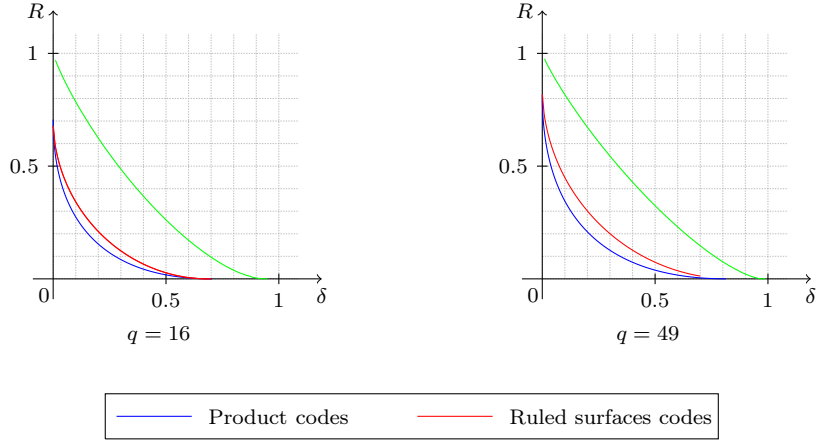
\begin{center}
\label{diagcodes}
\begin{multicols}{2}

\begin{tikzpicture}[x=3.0cm,y=3.0cm]
\draw [dash pattern=on .5pt off .5pt, xstep=0.3cm,ystep=0.3cm, lightgray] (0,0) grid (1.09,1.09);
\draw[->] (-0.09,0) -- (1.19,0);
\foreach \x in {0.5,1}
\draw[shift={(\x,0)}] (0pt,2pt) -- (0pt,-2pt) node[below] {\footnotesize $ \x$};
\draw[->] (0,-0.09) -- (0,1.19);
\foreach \y in {0.5,1}
\draw[shift={(0,\y)}] (2pt,0pt) -- (-2pt,0pt) node[left] {\footnotesize $\y$};
\draw[shift={(0,1.19)}] node[left] {\footnotesize $R$};
\draw[shift={(1.19,0)}] node[below] {\footnotesize $\delta$};
\draw[shift={(-.04,0)}] node[below] {\footnotesize $0$};
\draw[smooth,samples=100,domain=0:1] plot[parametric] function{t*t,(1-t)*(1-t)};
\draw[blue,smooth,samples=100,domain=0:1] plot ({36/51*\x*\x},{36/51*(1-\x)*(1-\x)});
\foreach \q in {16,49}
\draw[red,smooth,samples=100,domain=0.3:1] plot ({1-\x},{2/3+(\x+1)/34-1/51-2*sqrt((1/3-1/2-1/2*\x)*(\x-1)*18/34)});
\draw[green,smooth,samples=100,domain=0.01:0.95] plot ({\x},{1-(\x*ln(15)/ln(16)-\x*ln(\x)/ln(16)-(1-\x)*ln(1-\x)/ln(16)});
\draw[shift={(0.65,-0.25)}] node[left] {\footnotesize $q=16$};
\end{tikzpicture} \\

\begin{tikzpicture}[x=3.0cm,y=3.0cm]
\draw [dash pattern=on .5pt off .5pt, xstep=0.3cm,ystep=0.3cm, lightgray] (0,0) grid (1.09,1.09);
\draw[->] (-0.09,0) -- (1.19,0);
\foreach \x in {0.5,1}
\draw[shift={(\x,0)}] (0pt,2pt) -- (0pt,-2pt) node[below] {\footnotesize $ \x$};
\draw[->] (0,-0.09) -- (0,1.19);
\foreach \y in {0.5,1}
\draw[shift={(0,\y)}] (2pt,0pt) -- (-2pt,0pt) node[left] {\footnotesize $\y$};
\draw[shift={(0,1.19)}] node[left] {\footnotesize $R$};
\draw[shift={(1.19,0)}] node[below] {\footnotesize $\delta$};
\draw[shift={(-.04,0)}] node[below] {\footnotesize $0$};
\draw[smooth,samples=100,domain=0:1] plot[parametric] function{t*t,(1-t)*(1-t)};
\draw[blue,smooth,samples=100,domain=0:1] plot ({49/60*\x*\x},{49/60*(1-\x)*(1-\x)});
\draw[red,smooth,samples=100,domain=0.3:1] plot ({1-\x},{5/6+(3*\x+2)/300-sqrt((2/3+\x)*(1-\x)*51/50)});
\draw[green,smooth,samples=100,domain=0.01:0.99] plot ({\x},{1-(\x*ln(48)/ln(49)-\x*ln(\x)/ln(49)-(1-\x)*ln(1-\x)/ln(49)});
\draw[shift={(0.65,-0.25)}] node[left] {\footnotesize $q=49$};
\end{tikzpicture}

\end{multicols}
\begin{center}
\begin{tikzpicture}
\draw[blue,-] (0,0) -- (1,0);
\draw[shift={(1.1,0)}] node[right] {\footnotesize Product codes};
\draw[red,-] (4,0) -- (5,0);
\draw[shift={(5.1,0)}] node[right] {\footnotesize Ruled surfaces codes};
\draw[-] (-0.15,-0.3) -- (-0.15,.3) -- (8,.3) -- (8,-.3) -- (-0.15,-0.3);
\end{tikzpicture}
\end{center}
\end{center}

\section{Locality and avalaibility of the codes}

The last decade has seen an increasing interest on codes having local properties. In many situations, it is desirable to recover one coordinate of a vector in the code from a little number of other ones; this leads to the notions of \emph{locally decodable} or \emph{locally recoverable} codes and their study. Hirzebruch surfaces (or their images by the contraction of their negative curve) have already been used to construct such codes \cite{lana2,svav}.

We show below that the restrictions of the evaluation codes on ruled surfaces to certain sets of coordinates are well known: they are (subcodes of) Reed-Solomon codes, or algebraic geometric codes on the base curve. The properties of these last codes imply the local properties of the evaluation codes on ruled surfaces. We stress here on locally recoverable codes, and we give a very general result.

As usual, $\pi:X=\P(\E)\rightarrow C$ denotes a ruled surface, with $\E$ a locally free sheaf of rank $2$ over $C$. We fix a divisor $D\sim \O_X(a)+\pi^\ast \beta\in \Pic X$ with $\beta\in \Pic C$, $\deg \beta=b$. We denote by $\C$ the code obtained by evaluating the global sections of $D$ at the rational points of $X$.

We first describe the restriction of divisors on $X$ to certain curves.

\begin{lemma}
Let $\pi:X=\P(\E)\rightarrow C$ denote a ruled surface over $C$.

\begin{itemize}
\item[(1)] let $\imath : \P^1\rightarrow X$ denote a fiber of $\pi$; then the pull-back $\imath^\ast$ is the morphism from $\Pic X$ to $\Pic \P^1=\Z$ sending $\O_X(a)+\pi^\ast \beta$ to $\O_{\P^1}(a)$;
\item[(2)] let $s : C\rightarrow X$ denote a section of $\pi$ having image $S$ such that $S\sim \O_X(1)+\pi^\ast \gamma$; then the pull-back $s^\ast$ is the morphism from $\Pic X$ to $\Pic C$ sending $\O_X(a)+\pi^\ast \beta$ to $a(\det\E+\gamma)+\beta$.
\end{itemize}
\end{lemma}

\begin{proof}
The first assertion comes from the intersection product: the pull-back gives the restriction of the divisor $\O_X(a)+\pi^\ast \beta$ to a fiber of $\pi$. This restriction has degree $a$, and it must be $\O_{\P^1}(a)$.

We turn to the second assertion. We clearly have $s^\ast \pi^\ast \beta=\beta$ for any $\beta\in \Pic C$ since $s$ is a section of $\pi$. We compute the pull-back of $\O_X(1)$ by $s$.

We deduce from \cite[(II.7.12)]{hart} that we can associate to the morphism $s$ a surjection $\E\rightarrow \L$ where $\L=s^\ast\O_X(1)$ is an invertible sheaf. Let us write $\L=\O_C(\theta)$ for some divisor $\theta$ on $C$. On one hand, we know that the kernel of the above surjection is $\M\simeq\O_C(\det\E-\theta)$. On the other hand, we know from \cite[(V.2.6)]{hart} that $\pi^\ast \M\simeq \O_X(1)\otimes\O_X(-S)$. We get $S\sim \O_X(1)+\pi^\ast (\theta-\det\E)$, and finally $\theta\sim \gamma+\det\E$. This is the desired result.
\end{proof}

As a consequence, we deduce the

\begin{corollary}
\label{rest}
Let $\C$ denote the code obtained by evaluating the global sections of $D\sim \O_X(a)+\pi^\ast \beta$ at the rational points of $X$. Then

\begin{itemize}
\item[(1)] the restriction of $\C$ to the rational points on the fiber $\pi^{-1}(p)$ above $p\in C(k)$ is a subcode of the projective Reed-Solomon code of degree $a$, $\PRS(a)$. It is the whole projective Reed-Solomon code when we have $H^1(X,\O_X(D-\pi^\ast p))=0$.
\item[(2)] if $s : C\rightarrow X$ denotes a section of $\pi$ having image $S$ such that $S\sim \O_X(1)+\pi^\ast \gamma$; then the restriction of $\C$ to the rational points on the curve $S$ is a subcode of the code obtained by evaluating the global sections of $a(\det\E+\gamma)+\beta\in \Pic C$ at the points of $C(k)$. The restriction is surjective  when we have $H^1(X,\O_X(D-S))=0$.
\end{itemize}
\end{corollary}

\begin{proof}
We just show the first assertion from the first part of the preceding lemma; the proof of the second one is entirely similar.

Let us write the exact sequence that defines the subscheme $\pi^{-1}(p)$ of $X$ :
\[
0\rightarrow \O_X(-\pi^\ast p)\rightarrow \O_X\rightarrow \O_{\pi^{-1}(p)}\simeq\O_{\P^1}\rightarrow 0
\]
We tensor it by $\O_X(D)$; the right hand sheaf becomes $\O_{\P^1}(a)$ from the above Lemma. Now we consider the long exact sequence obtained by taking global sections, and we get a map from $H^0(X,\O_X(D))$, the set of global sections of $D$, to $H^0(\P^1,\O_{\P^1}(a))$ that we identify with the space of homogeneous polynomials of degree $a$ in two variables. Moreover, this map is surjective when the cohomology space $H^1(X,\O_X(D-\pi^\ast p))$ vanishes.
\end{proof}

Let us define the codes of interest to us in this section

\begin{definition}
A code $\C$ of length $n$ is \emph{locally recoverable} with \emph{locality} $r$ and \emph{availability} $s$ when for every $1\leq i\leq n$, there exist $s$ pairwise disjoint subsets $J_{i1},\ldots,J_{is}$ of $\{1,\ldots,n\}$, each of cardinality $r$ (the \emph{recovery sets}), such that for any codeword $c\in\C$, we can recover its coordinate $c_i$ from any of the sets $\{c_j,~j\in J_{ik}\}$, $1\leq k\leq s$. 
\end{definition}

It follows immediately from ``projective'' Lagrange interpolation that the projective Reed-Solomon code $\PRS(a)$ over $\F_q$ is locally recoverable with locality $a+1$ and avalaibility $\lceil\frac{q}{a+1}\rceil$. As usual we deduce a polynomial of degree at most $a$ from $a+1$ of its values, the only difference with classical Lagrange interpolation being that we also consider the values of such polynomials at infinity; this is just their degree $a$ coefficient.

From this result we deduce the following local properties of codes on ruled surfaces

\begin{proposition}
Let $D\sim \O_X(a)+\pi^\ast \beta\in \Pic X$ with $\beta\in \Pic C$ denote a divisor on a ruled surface $\pi:X=\P(\E)\rightarrow C$ defined over $\F_q$. We assume that for any $p\in C(\F_q)$ we have $H^1(X,\O_X(D-\pi^\ast p))=0$.

Then the code $\C$ obtained by evaluating the global sections of $D$ at the rational points of $X$ is locally recoverable with locality $a+1$ and avalaibility $\lceil\frac{q}{a+1}\rceil$.
\end{proposition}

\begin{remark}
Since $a$ is positive, the condition $H^1(X,\O_X(D-\pi^\ast p))=0$ can be reduced to a condition of vanishing of some cohomology spaces on the base curve (ie to some non speciality conditions) from \cite[(V.2.4)]{hart}.

For instance, if we have $\E=\O_C\oplus\O_C(\varepsilon)$, this condition is equivalent to the divisors $\beta+i\varepsilon-p$ being non special for any $0\leq i\leq a$ and any $p\in C(\F_q)$. This last condition is verified when the degree of $\beta$ is large enough.
\end{remark}

\begin{remark}
We have thus shown that under a mild assumption, our codes have ``fibral'' recovery sets. But we could also use the second part of Corollary \ref{rest} to prove that under certain conditions, the code has also ``sectional'' recovery sets, increasing its avalaibility. It is clear that an evaluation code on a ruled surface inherits many local properties from those of certains codes on the base curve. We leave it to the interested reader to choose his favorite curve as base curve, and deduce the corresponding local properties on the ``global'' codes.
\end{remark}

%
%

\bibliographystyle{smfplain}

\bibliography{DoubleCover}

\end{document}